\documentclass{article}
\usepackage[sfdefault]{roboto}

\usepackage[english]{babel}
\usepackage{authblk}
\usepackage{amsmath}
\usepackage{amsthm}
\usepackage{bbm}
\usepackage{url}
\usepackage{hyperref}
\usepackage{epigraph}
\usepackage[T1]{fontenc}

\usepackage[a4paper,top=2cm,bottom=2cm,left=3cm,right=3cm,marginparwidth=1.75cm]{geometry}
\usepackage{xcolor}
\usepackage[export]{adjustbox}

\usepackage{graphicx} % Required for inserting images

\usepackage{tikz}
\usetikzlibrary{arrows.meta}
\tikzset{mynode/.style={draw, very thick, circle, minimum size=1.1cm},
    myarrow/.style={very thick, -Triangle}}

\newtheorem{definition}{Definition}
\newtheorem{example}{Example}
\newtheorem{fact}{Fact}
\newtheorem{corollary}{Corollary}
\newtheorem{theorem}{Theorem}

\newcommand{\N}{\ensuremath{N}} %agents
\newcommand{\D}{\ensuremath{D}} %delegation profile
\newcommand{\CalD}{\ensuremath{\mathcal{D}}} %delegation profile
\renewcommand{\P}{\ensuremath{\mathbbm{P}}} %probability operator
\newcommand{\E}{\ensuremath{\mathbbm{E}}} %expectation operator
\newcommand{\set}[1]{\left\{ #1 \right\}}
\newcommand{\tuple}[1]{\left\langle #1 \right\rangle}
\newcommand{\1}{\ensuremath{\mathbbm{1}}}

\begin{document}
%{\fontfamily{lmss}\selectfont

\title{
Delegations as Adaptive Representation Patterns:
Rethinking Influence in Liquid Democracy}
\author[1]{Davide Grossi}
\author[2]{Andreas Nitsche}
\affil[1]{University of Groningen and University of Amsterdam}
\affil[2]{Association for Interactive Democracy Berlin}% and Project LiquidFeedback
\date{October 2026}

%%%%%%%%%%%%%%%%%%%%%%%%%%%%%%%%%%%%%%

\maketitle

\epigraph{Even in reasoning upon some subjects, it is a mistake to aim at an unattainable precision. It is better to be vaguely right than exactly wrong.}{\textit{Carveth Read \\ Logic: Deductive and Inductive (1898)}
}

\begin{abstract}
Liquid democracy is a mechanism for the division of labor in decision-making through the transitive delegation of influence. In essence, all individuals possess the autonomy to determine the issues with which they will engage directly, while for other matters, they may appoint a representative of their choosing.
So far, the literature has studied the delegation structures emerging in liquid democracy as static. As a result, transitivity defined as the capacity to transfer acquired authority to another entity, has been identified as a concern as it would be conducive to unrestrained accumulation of power.

Focusing on the implementation of liquid democracy supported by the LiquidFeedback software, we propose a novel approach to assessing the influence of voting nodes in a transitive delegation graph, taking into account the process nature of real-world liquid democracy in which delegation and voting are distinct and increasingly independent activities. By introducing a novel model of delegations in liquid democracy, we show how transitivity may in fact contribute to an effective regulation of deliberation influence and decision-making power. 
While maintaining the one-person, one-vote paradigm for all votes cast, the anticipated influence of an agent, to the extent it is stemming from transitivity, experiences a precipitous decline following an exponential trajectory. In this context, the emergence of representation patterns will be examined, the topic of delegation cycles will be revisited, and the notion of gurus will be challenged.
While this approach does not claim universal validity, it introduces a crucial perspective that is indispensable for evaluating the ramifications of implementation decisions, measures intended to curtail power, and alternative delegation models. 
%It will be exemplified how interventions with commendable intentions can yield unintended consequences.

In general, it is our objective to move the first steps towards a rigorous analysis of liquid democracy as an adaptive democratic representation process. The adaptivity aspect of liquid democracy has not yet been explored within the existing academic literature despite it being, we believe, one of its most important features. We therefore also outline a research agenda focusing on this aspect of liquid democracy.

\bigskip
{\bf Keywords:} Liquid Democracy, Transitive Delegation, Delegation Graph, Collective Intelligence Systems, Delegation Loops, Networked Representation, LiquidFeedback

\end{abstract}

%\begin{figure}%[htp]
    %\includegraphics[width=0.5\textwidth, right]{images/EN_FundedbytheEU_RGB_POS.png}
%\end{figure}

\newpage
\tableofcontents
%\newpage
\section{Introduction}

\paragraph{Liquid Democracy.}
When direct member influence is implemented within an organization, it inevitably gives rise to several critical questions. Primarily, it is essential to ascertain whether all constituents possess a desire to engage with every matter that comes before the body. Moreover, it is important to consider the possibility that individuals may harbor divergent interests, leading to a diverse array of preferences regarding the matters that they wish to have a direct voice in. 

Liquid democracy has the potential to provide an adaptive solution to these issues \cite{behrens2014principles,blum2016liquid,valsangiacomo2021political,valsangiacomo2025six}. In essence, liquid democracy enables direct participation whenever deemed appropriate by an individual member, without imposing undue burden. Liquid democracy is predicated on the notion of individual freedom of choice, whereby members are empowered to select the topics with which they wish to engage and may adjust these preferences over time, consequently shaping their own position within the two-dimensional [topic x time] continuum that delineates direct and representative democracy. In the extreme, members can exercise autonomy in all matters without the need for delegation, thereby embodying a direct democratic spirit. Conversely, members have the option to delegate their rights comprehensively, thus adopting a representative democratic framework. The efficacy of liquid democracy is most evident in contexts with a wide thematic range where the division of labor is pronounced, such as in the governance of organizations, including key democratic institutions.\footnote{It may be worth noting at the outset that civic participation does not necessarily align with this criterion. Frequently, civic engagement is characterized by a monothematic nature or a limited time frame and scope. Conversely, if a civic participation system is regarded as an ongoing democratic infrastructure, liquid democracy can assume an instrumental role.}

\paragraph{Aim of the paper.} The last decade has witnessed a growing literature on liquid democracy: from political science \cite{blum2016liquid,valsangiacomo2021political,valsangiacomo2025six} to, in particular, (computational) social choice \cite{brandt2016handbook}. 
The latter strand of literature has been very productive in developing a variety of models in an attempt to analyze the standard implementation of liquid democracy, and of several variants thereof (see \cite{papasotiropoulos2025liquid} for a recent extensive overview of such literature).
A subsequent survey systematizes this literature by developing a taxonomy of liquid democracy models within computational social choice \cite{grossi2026modeling}.
Despite the many insights provided by this literature, many aspects of even the standard implementation of liquid democracy have remained either hidden or only partially touched upon. In this paper we aim at bringing such limitations into light by insisting on the need for models of liquid democracy that can do justice to its adaptivity. 
We illustrate the benefits that such a modeling move can enable by proposing a novel and, we argue, better informed understanding of influence in the presence
of transitive delegations. 

We strongly believe that targeted research in this area could provide more solid foundations for existing software implementations of liquid democracy systems, as well as stimulate evidence-based improvements to the current deployments of such systems. 

\paragraph{Special Focus: Transitive Delegations.}
In this paper we will focus on the standard `prototypical' implementation of liquid democracy, as realized for instance in the LiquidFeedback software.\footnote{\url{https://liquidfeedback.com/en/}}.
%The system under consideration is a dynamic transitive liquid democracy system with 
in which agents that receive delegations from \emph{none, one, or multiple} agents and forward their own and all received delegations to \emph{none or one} agent. A voting agent casts their own vote and all votes received via delegation. In the case of a non-voting agent, the agent in question relays their own vote and all votes received by delegations if an outgoing delegation is in effect.

The fundamental components of liquid democracy are \cite{behrens2014principles}: (1) direct participation, (2) representation through delegation, (3) topic specific delegations, (4) transitivity of delegations, (5) change of delegations, (6) instant recall of delegations. Constituents possess the capacity to exercise direct democratic rights through voting or to delegate their voting rights to another individual. The delegation of voting rights can be constrained to specific subject matter, and these delegations can be further transferred to other individuals. Additionally, these delegations can be modified or recalled at the discretion of the delegator.

In this paper, we specifically focus on the fourth of the above components: transitivity of delegations. This principle enables a flexible and scalable delegation mechanism through trust-based networks. By enabling the passage of votes through multiple delegation steps, liquid democracy fosters a hybrid model that integrates the inclusivity of direct participation with the efficiency of representative systems. One of the primary advantages of transitivity lies in its capacity to empower individuals to delegate voting authority to peers, rather than to public figures or institutional elites. It is important to note that this delegation does not necessitate that the peer be an expert or even be involved in the decision-making process. Instead, the peer has the option to further delegate to an agent who possesses greater knowledge or engagement. In this manner, transitivity fosters the development of organic representation within a decentralized framework. The\emph{ one-person, one-vote} paradigm is maintained for all votes cast. 
%This system structure is implemented by systems such as the abovementioned LiquidFeedback. 

\paragraph{Contribution and Structure of the Paper.} 
After having reviewed what we consider to be core features of liquid democracy (Section \ref{sec:LF}), we illustrate (Section \ref{sec:expected}) the importance of giving the adaptivity aspect of liquid democracy a proper position in liquid democracy models. We do so by developing a novel theory of influence where transitive delegations emerge as an important feature to modulate power in the system, rather than compound its accumulation. The mathematics of the model is developed in more details in the technical appendix, which we consider integral part of the paper. Finally, in Section \ref{sec:outlook} we discuss a number of directions of research which we believe to be of special importance to improve our understanding of liquid democracy and, therefore, also its implementation in software and its deployment in practice.

%%%%%%%%%%%%%%%%%%%%%%%%

\section{
Liquid Democracy in LiquidFeedback
%-style Implementation
} \label{sec:LF}

In this section, we recollect what we considered to be the key features of liquid democracy in the form currently implemented in the LiquidFeedback software.

\subsection{Beyond Voting: Deliberation and Decision Making}

Liquid democracy is a system that provides a general mechanism for the division of labor through the delegation of influence. This influence is often interpreted as voting; however, it can actually be applied to any quantifiable activity. This aspect has not yet been addressed in academic literature.

In the LiquidFeedback  \cite{behrens2014principles} implementation of liquid democracy, such activities include assigning greater weight to a person's argument or support in a deliberation. Consequently, this also increases an agent's influence in collective moderation schemes and in identifying viable voting options. In the subsequent discussion, the term "voting" is employed to denote all quantifiable activities. It is imperative to recognize the significance of this matter, as the capacity to contribute to deliberation, and more specifically, the development of viable voting options, is of equal importance to the act of casting votes.

\subsection{Evolving Delegation Graph: Domain Specific Representation Patterns}

Liquid democracy has been conceptualized as a system that aims to alleviate the workload on agents. This objective would be unattainable if the delegation structure had to be established from the ground up with each new decision, as it would necessitate every participant to make a delegation decision on every subject. Consequently, LiquidFeedback resorts to delegation structures that evolve over time, thereby creating a representation pattern on how to address new decisions.  

A hierarchy of delegations, an overlay of default delegations, policy area-specific defaults and topic-specific delegations, give rise to multiple domain-specific representation patterns.\footnote{This delegation hierarchy does not constitute an integral component of the prototypical implementation. However, it is fully compatible with the it. A consolidated delegation graph is in effect for a given topic at any given time, in which the most specific delegation per agent for the topic in question is applied. This consolidated graph fulfills all the assumptions concerning the delegation graph that were outlined for the prototypical reference implementation above.} Agents can delegate once, have their preferences apply broadly, and refine as they go. This further alleviates the burden on participants and enhances the practicability of liquid democracy.

Rather than reconfiguring from scratch for every issue, the system builds and refines domain-specific representational patterns over time. The representation patterns delineate the weight of a potential agent activity. These represent an epistemic specialization among agents, referring to potential weight upon action. Agents who demonstrate consistent proficiency in specific domains are likely to accrue greater potential influence. However, this does not influence the ability to engage directly for any other agent.

%There is an overlay of default, default by policy area, and topic-specific delegations. These lead to the emergence of decision-making structures that are adapted according to the policy area and, if necessary, refined for specific topics. 

%However, a consolidated delegation graph is in effect for a given topic at any given time, in which the most specific delegation per agent for the topic in question is applied. This consolidated graph fulfills all the assumptions concerning the delegation graph that were outlined for the prototypical reference model.

%It is possible to comprehend the representation patterns in liquid democracy as a form of cognitive architecture. This architecture enables the system to scale efficiently without overwhelming participants. By facilitating the evolution of delegation structures, liquid democracy fosters an implicit organizational memory—that is, a system that encodes experience, expertise, and trust relationships into the very fabric of decision-making.

%These pathways function analogously to heuristics in the human brain, facilitating streamlined decision-making by leveraging prior learning and established trust networks.

%\bigskip
%\textit{ Example: 

\begin{example}
In the context of a LiquidFeedback system, a constituent of an organization is endowed with the capacity to appoint another member of the organization to act as their proxy on all matters that the appointing member does not engage themselves. Furthermore, the appointing member can override this choice by appointing another member to represent them in a particular policy domain or on a specific subject. In either scenario, the member retains the capacity to act directly, which results in the automatic suspension of the delegation for the particular activity.
\end{example}

\begin{figure}[t]
    \centering   \includegraphics[width=1\linewidth]{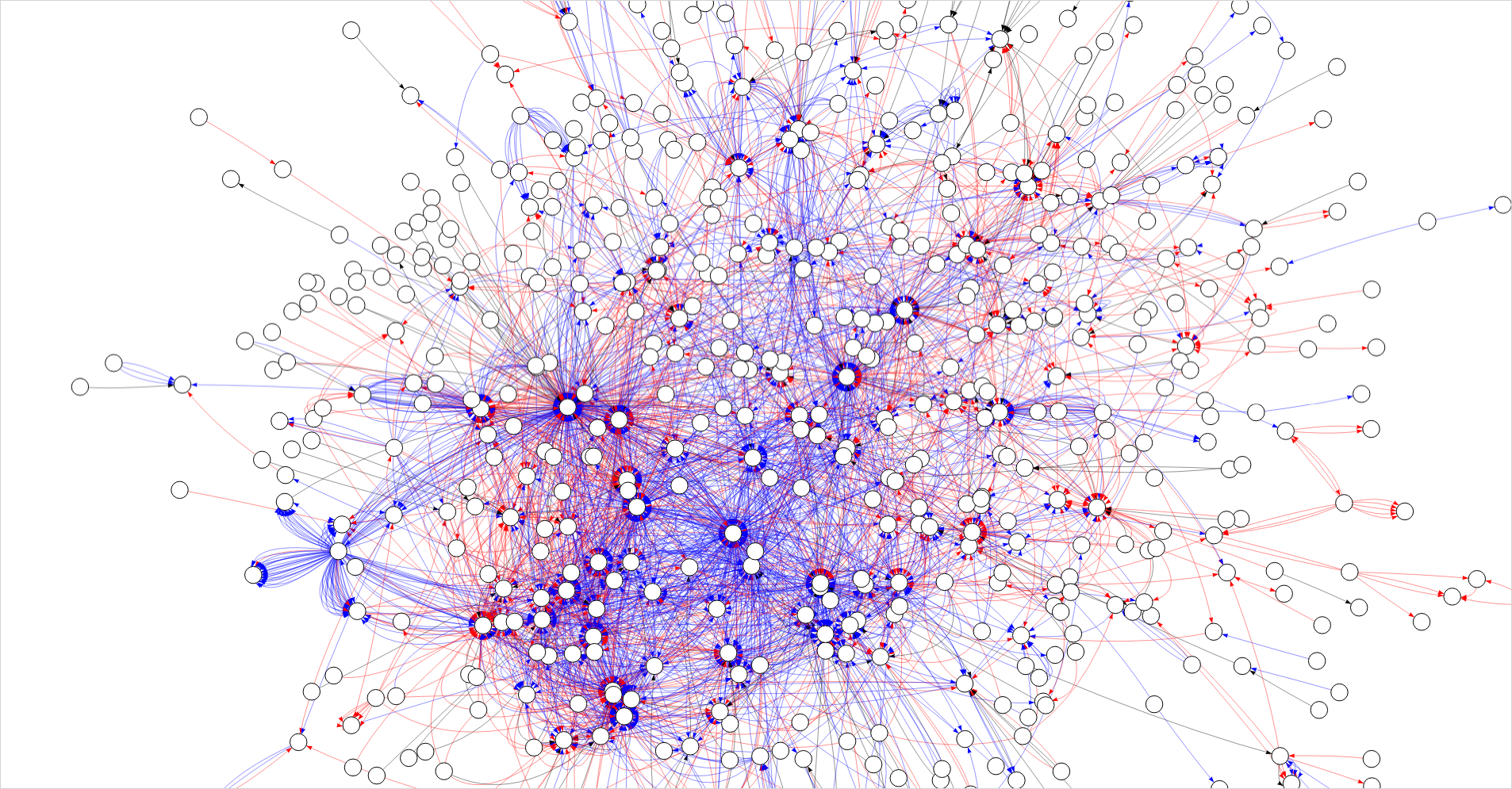}
    \caption{\emph{LiquidFeedback Delegation Overlay} 
    Default delegations %for the enclosing organizational unit 
    (black), default delegations by policy area (red) as well as issue specific delegations (blue)}
\label{fig:Delegation_Overlay}
\end{figure}

\subsection{Delegation Features: Recall and Suspension, Defaults}

\paragraph{Recall and Suspension.}

Agents possess the capacity to execute any quantifiable activity pertaining to a particular issue in a direct manner. In the event that this occurs, it is imperative that the outgoing delegation is not in effect for that specific issue and activity. Theoretically, this could be achieved by recalling a delegation (instant recall); however, such an action would be impractical and would undermine the division of labor. Instead, LiquidFeedback suspends delegations for any activity a participant directly exercises for a given issue.

Both recall and suspension exist in LiquidFeedback. While recalling a delegation constitutes an active action, suspending a delegation is passive in nature and is triggered by another activity, such as discussion, support, or voting.\footnote{According to the established reference model, the suspension of a delegation is a viable course of action. This is due to the model's design, which allows for the recalling and subsequent renewal of a delegation at a later point in time.} 

In the event that an agent has lost confidence in their proxy, the most rational course of action would be to initiate a recall. However, if the agent's objective is merely to contribute to a discussion, this step can be omitted. The implementation of an automatic suspension ensures that the delegation is in full effect for the subsequent voting on the designated issue, as well as for the discussion and voting on any other issue that falls within the scope of this delegation.

%\subsection{Delegations in Abundance}

%\paragraph{Defaults.} 
%The determination of delegation is made from the perspective of an agent who is tasked with answering the fundamental question, "Who, if anyone, is to act on my behalf in the event that I do not act myself?" %Depending on the implementation, this choice may be further refined, e.g., for policy areas.

\paragraph{Defaults.} The determination of delegation is made from the perspective of an agent who is tasked with answering the fundamental question of who, if anyone, is to act on their behalf in the event that they do not act themselves.
%\paragraph{Preemptive Delegations.} %alternative: Advance Delegations
There is no requirement for agents to make an immediate decision regarding their direct engagement in the deliberation process of an issue and—as an independent decision—whether to cast their vote directly. Indeed, in the case of default or default by policy area delegations, these will be applicable to future issues that are not yet known to the agent at the time the delegation is created. For certain issues, an agent may already have the intention to act directly or to be represented. Nevertheless, even in these cases, there may be circumstances that lead to a deviation from the intention.

Unless agents strongly oppose representation, for example, as direct democracy proponents, it is rational to establish backup delegations for all matters by designating a default delegation and potentially refining it. Delegating agents always retain the capacity to act directly, which results in the automatic suspension of the delegation for the particular activity. 

In a LiquidFeedback-style implementation, the delegation graph contains no actionable information regarding the potential actions or voting intentions of an agent.

\subsection{Delegation Structures: Endpoints and Gurus, Cycles}

\paragraph{Endpoints and Gurus.} Agents devoid of an outgoing delegation are designated as endpoints, insofar as they lack the capacity to transfer their own vote or any vote received through delegation. These endpoints are often accentuated as being of particular significance and are sometimes referred to as "gurus" if they receive a significant number of incoming delegations, presuming that they are the agents which will effectively cast votes. Nevertheless, from a practical standpoint, there is no substantiation to support this assumption. 

The absence of an outgoing delegation is no commitment to exercise the right to vote. Conversely, the presence of an outgoing delegation is no indication of an agent's intent to abstain from voting. In most cases, default and default by policy area delegations are established prior to the creation of an issue. A weak correlation may exist, particularly in instances of issue-specific delegations, which are the lowest level in the delegation hierarchy and are generally established posterior to the emergence of an issue.

Subsequently, it will be demonstrated that the influence of an agent is contingent on topographical factors independent of that agent's position as an endpoint.

\paragraph{Delegation Cycles.}
A delegation cycle occurs when the delegation choice of an agent and the choices of their delegation choice lead back to this agent. They have either no effect or are automatically resolved by the suspension mechanism if any of the agents in the cycle votes.  Cycles may occur inadvertently, particularly protracted cycles, and are at times deliberately engineered. The deliberate cycles are ordinarily brief cycles within the same level, predominantly within the default or default by area level. 

We examine two exemplary cycles: \textbf{1.} A mutual delegation \( a\) → \(b\) → \( a \). It is hereby asserted that, without loss of generality, \( a \) casts a vote, thereby bringing the cycle to a resolution: \( b\) → \( a \). It is possible that both \( a \) and \( b\) have received additional delegations from outside the cycle. \( a \) votes with their own vote, the vote of \( b\) and all delegations to \( a \) and \( b\) from outside the cycle. Observe that \( b\) can still vote in which case the delegation to \( a \) is suspended. In this case, both agents are decoupled and each agent individually casts their own votes and delegations, each of them has received from other agents. \textit{\textbf{2.} \( a \) → \( b\) → \( c\) → \( a \). }It is hereby asserted that, without loss of generality, \( a \) casts a vote, thereby bringing the cycle to a resolution: \( b\) → \( c\)  → \( a \). It is possible that \( a \), \( b\) and \(c\) have received additional delegations from outside the cycle. This is illustrated in Figure \ref{fig:cycle}.

\begin{figure}[t]
\begin{center}
\begin{tikzpicture}
\node[mynode](n0) at (0,0){3};
\node[mynode](n1) at (-2,-3){3};
\node[mynode](n2) at (2,-3){3};
\node[mynode](n3) at (7,0){3};
\node[mynode](n4) at (5,-3){1};
\node[mynode](n5) at (9,-3){2};
\draw[myarrow](n0)--(n1);
\draw[myarrow](n1)--(n2);
\draw[myarrow](n2)--(n0);
\draw[myarrow](n4)--(n5);
\draw[myarrow](n5)--(n3);
\draw[myarrow](n3)--(4.5,0);
\end{tikzpicture}
\end{center}
\caption{\emph{Delegation cycle} (left)  vs. \emph{resolved cycle} after initial vote cast}
\label{fig:cycle}
\end{figure}
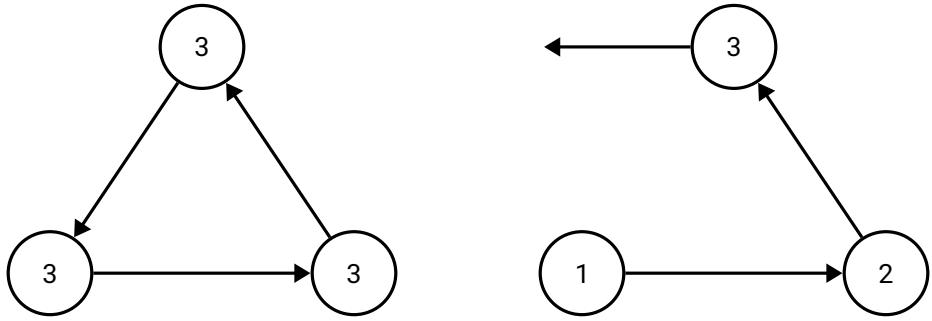

\bigskip
All agents in the cycle can be interpreted as a compound endpoint with every agent within the cycle able to initially trigger all delegations reaching the cycle, i.e. delegations within the cycle and delegations reaching the cycle from outside the cycle. 

From the perspective of an agent, deliberate cycles can be employed as a mechanism to mitigate the risk of forgetting to cast a vote. It is evident that cycles have no negative impact on the number of votes cast; rather, they increase the probability of more votes cast. Cycles are beneficial at most.

%In summary, if maximizing votes is considered advantageous, cycles are beneficial at most.

\section{
A Process-Centered Model of Influence
%Expected Voting Weight
} \label{sec:expected}

In this section we suggest a novel model of influence in liquid democracy, which is able to capture, albeit in a stylized fashion, the key aspects of delegations we emphasized in the previous section. In particular it captures explicitly the possibility of agents transitioning between the two modes of representation available in liquid democracy: direct voting, and representation via proxies. The exposition in this section favors intuitions over formalism. A more technical exposition of the content, including further references to existing literature, is provided in the technical appendix.

\subsection{Assuming All Agents Vote with a Probability of 50\% }

The current understanding of influence in liquid democracy is based on a simple accrual idea \cite{christoff2017binary,golz2021fluid}: the weight of an agent amounts to the number of other agents that delegate directly or indirectly to them. We refer to this notion as {\em potential weight}. See figure \ref{fig:potential_expected} (left) for an illustration.

Our objective here is to develop an alternative understanding of the process through which influence is accumulated within a delegation system. To that end, we endeavor to calculate the expected number of votes a agent is likely to cast. We initially operate under the simplifying assumption that all agents cast a vote with an equal probability of 50\%.

From the perspective of a voting agent \( t \), we want to determine how many total votes they cast when they vote, considering the delegation process where \( t \) represents an agent in a transitive delegation voting system at step t.

\begin{itemize}
    \item Each agent other than \( t \) votes directly with 50\% probability.
    \item If an agent does \textbf{not} vote, they transfer their own vote and all votes received through delegation to their chosen delegate, if such a delegate exists.
    \item An agent that votes casts their own vote plus any votes they have received through delegation.
\end{itemize}

If \( t \) votes, they cast:  
\begin{enumerate}
    \item their own vote (always 1).
    \item All votes delegated to them from other agents that chose not to vote and successfully transferred their votes down the chain.
\end{enumerate}

For any agent \( i \) that delegates to \( t \), their vote reaches \( t \) only if:
\begin{itemize}
    \item \( i \) does not vote (50\% probability).
    \item The delegation continues through any intermediate agents without interruption.
\end{itemize}

Since each agent only transfers their own vote and any received delegations if they do not vote, the probability of a vote reaching \( 
\) from an agent at delegation distance \( d \) is:
\begin{equation}
(0.5)^d
\end{equation}
where \( d \) is the number of delegation steps before the vote reaches \( t \).

The total expected number of votes \( t\) casts when they vote is:

\begin{equation}
\E[V_t] = 1 + \sum_{i \in S} (0.5)^{d_i}
\end{equation}
where:
\begin{itemize}
    \item \( V_t \) represents the total number of votes cast by agent \( t \) when they vote.
    \item \( S \) is the set of all agents that may eventually delegate to \( t \).
    \item \( d_i \) is the delegation distance from agent \( i \) to \( t \). 
\end{itemize}

\medskip

%\paragraph{Special Cases}
We now illustrate further the above model considering two extreme cases of delegation structures: linear delegation chains, and direct delegations (all delegation paths are of length $1$).

\paragraph{Linear Chain of Delegations.}
   If \( t \) is at the end of a single delegation chain of length \( n \), the sum is a finite geometric series, which simplifies to:
   \begin{equation}
   \E[V_t] = 2 \times (1 - (0.5)^{n+1})
   \end{equation}
   For large \( n \), the expected number of votes approaches $2$.

%\bigskip
%\textit{ Example: 
\begin{example}
In accordance with the one-person-one-vote principle, an agent that has received $10000$ delegations via a linear delegation chain can initially cast $10001$ votes (i.e., their own vote and the 10,000 delegations). Given a probability of $50\%$ for agents to cast a direct vote, the expected number of votes that ultimately will be cast by this agent is approximately 2.
\end{example}

\paragraph{Direct Delegations.}
   If \( t \) receives k direct delegations, the sum increases accordingly:
   \begin{equation}
   \E[V_t] = 1 + \frac{k}{2}
   \end{equation}
The expected number of votes cast grows linearly with the number of incoming direct delegations.

%\bigskip
%\textit{ 
\begin{example} \label{ex:direct}
In accordance with the one-person-one-vote principle, an agent that has received $10000$ delegations via all direct delegations can initially cast $10001$ votes (i.e., their own vote and the 10,000 delegations). Given a probability of $50\%$ for agents to cast a direct vote, the expected number of votes that ultimately will be cast by this agent is $5001$.
\end{example}

\subsection{Generalization to Arbitrary Uniform Probabilities}

In this section we generalize the intuitions outlined in the previous section to a general model of influence in liquid democracy. The model is further studied in the technical appendix.

\paragraph{Expected Number of Votes Cast by \( t \).}
 
%From the perspective of a voting agent \( t \), we want to determine how many total votes it casts when it votes, considering a delegation process where:  

We now assume that each agent votes directly with probability \( p \) (instead of a fixed $50\%$).

%\begin{itemize}
%    \item Each agent votes directly with probability \( p \) (instead of a fixed 50\%).
%    \item If a agent does \textbf{not} vote, it delegates their own vote and all votes received through delegation to their chosen delegate, if such a delegate exists.
%    \item An agent that votes casts their own vote plus any votes it has received through delegation.
%\end{itemize}
  
%If \( t \) votes, it casts:  
%\begin{enumerate}
%    \item their own vote (always 1).
%    \item All votes delegated to it from other agents that chose \textbf{not} to vote and successfully passed their votes down the chain.
%\end{enumerate}

For any agent \( i \) that delegates to \( t \), their vote reaches \( t \) only if:
\begin{itemize}
    \item \( i \) does not vote (probability \( 1 - p \)).
    \item The delegation continues through any intermediate agents without interruption.
\end{itemize}

Since each agent delegates only if it does not vote, the probability of a vote reaching \( t \) from an agent at delegation distance \( d \) is:

\begin{equation}
(1 - p)^d
\end{equation}

where \( d \) is the number of delegation steps before the vote reaches \( t \).

The total expected number of votes \( t \) casts when it votes is:

\begin{equation}
\E[V_t] = 1 + \sum_{i \in S} (1 - p)^{d_i} \label{eq:expect}
\end{equation}

where:
\begin{itemize}
    \item \( V_t \) represents the total number of votes cast by agent \( t \) when they vote.
    \item \( S \) is the set of all agents that may eventually delegate to \( t \).
    \item \( d_i \) is the delegation distance from agent \( i \) to \( t \).
\end{itemize}
In the technical appendix we will be more precise about the type of expectation that is built into Equation \eqref{eq:expect}.

Subsequent work has provided a network-centrality characterization of this
notion of influence. Expected voting weight in the model of suspendible
delegations has been shown to coincide with Random Walk Decay centrality,
together with an axiomatic and computational characterization of the resulting
power measure \cite{grossi2026centrality}.

It is important to observe right away that the expected number of votes cast by an agent in a cycle can be calculated in the same way as anywhere else in the delegation graph with the understanding that the outgoing delegation of \( t \) is not taken into account, as the assumption for calculating  E[\(V_t] \) is that \( t \) casts a vote and consequently the outgoing delegation is suspended.

%\paragraph{Special Cases}

\medskip
As done above, we now consider two extreme cases of delegation structures: linear delegation chains, and direct delegations (all delegation paths are of length $1$).

\paragraph{Linear Chain of Delegations.}
   If \( t \) is at the end of a single delegation chain of length \( n \), the sum simplifies to:

   \begin{equation}
   \E[V_t] = 1 + \frac{1 - (1 - p)^{n+1}}{p}
   \end{equation}

   For \( n \to \infty \) this approaches:
   \begin{equation}
   %E[V_t] \approx 1 + \frac{1}{p} - 1 = \frac{1}{p}
   \E[V_t]  = \frac{1}{p}
   \end{equation}

Here are the expected number of votes cast for a \emph{long delegation chain} (\( n \to \infty \)) at different voting probabilities \( p \):  

\begin{center}
\begin{tabular}{ c c }
 \( p \) (Voting Probability)& Expected Votes \( \E[V] \) \\ 
 1\%  (\( 0.01 \))& 100\\  
 2\%  (\( 0.02 \))& 50\\ 
 5\%  (\( 0.05 \))& 20\\ 
 10\%  (\( 0.1 \))& 10\\ 
 20\%  (\( 0.2 \))& 5\\ 
 50\%  (\( 0.5 \)) & 2\\ 
 80\%  (\( 0.8 \)) & 1.25\\ 
\end{tabular}
\end{center}

%\textit{Example: 

\begin{example}
In a delegation chain, a $1\%$ probability for casting a vote is assumed, the expected number of votes cast would be $100$, which is less than $1\%$ in any liquid democracy system with more than $10000$ participants.
\end{example}

\paragraph{Direct Delegations.}

If \( t \) only receives direct delegations, the sum increases accordingly:
  \begin{equation}
   \E[V_t]=1+k  \times (1 - p)
   \end{equation}
The expected number of votes cast grows linearly with the number of incoming direct delegations. 

\subsection{Limitations and Implications of the Model}

In this section we articulate what we consider to be the key implications of the above model and emphasize some if its limitations.

\paragraph{Uniform probabilities.} The above model is certainly a simplification of how delegations and their suspension work in liquid democracy, first and foremost because it assumes that all agents vote with equal probability. This assumption is not granted in reality. Individual attachment to the topic, as well as factors such as available time, become pertinent. Moreover, the gravity of the decision influences the level of participation, with high-stakes decisions being characterized by a heightened level of participation.

%\subsection{Expected Votes in Delegation Cycles}

\paragraph{Power Accumulation and Transitive Delegations.}
\label{par:accumulation}

Despite the above simplification however, the model still allows us to make the key observation that transitivity does not play a significant role in the accumulation of power.
%, provided that intervention points, as observed in prototypical liquid democracy, exist. 
To the contrary, and perhaps counter to intuitions, direct delegations appear to be a more likely contributor to the accumulation of persistent power compared to transitive delegations. 
The phenomenon of transitivity, or the ability to pass on received power to another individual, actually appears to play a crucial role in maintaining a balance of power. By enabling the delegation of authority to peers rather than to prominent figures, transitivity fosters a more inclusive and egalitarian decision-making environment. Furthermore, it establishes intervention points that allow for the modification or even the reversal of power dynamics, thereby promoting stability and fairness within the system.  

We therefore observe that the imposition of limitations on transitive delegations may actually foster the accumulation of power, rather than contrasting it. In the context of a transitive graph characterized by extensive delegation chains, the anticipated influence experiences a precipitous decline, following an exponential trajectory. Limiting the number of hubs or imposing a cost on the transfer of delegations is likely to motivate direct delegation to the perceived voter, resulting in the aggregation of direct delegations and, as an unintended consequence, encouraging the accumulation of less constrained power. This exemplifies that even interventions with commendable intentions can yield unintended consequences. Such an understanding is critical for evaluating the ramifications of implementation decisions, measures intended to curtail power, and alternative delegation models.

It is hereby suggested that the expected number of votes serve as an effective indicator of sustained power accumulation. Figure \ref{fig:potential_expected} illustrates the above intuitions.
We select again uniform probability for direct voting of 50\%, obtaining 
%\begin{equation}
$
\E[V_t] = 1 + \sum_{i \in S} (0.5)^{d_i}.
$
%\end{equation}
The figure provides a visual comparison of the \emph{potential voting weight} and the \emph{expected voting weight}. The values for the expected number of votes were rounded to the third decimal place where applicable.

\begin{figure}[t]
\scalebox{0.6}{
\begin{tikzpicture}

\node[mynode](n0) at (0,0){7};
\node[mynode](n1) at (160:7){9};
\node[mynode](n2) at (180:5){9};
\node[mynode](n3) at (140:5){9};
\node[mynode](n4) at (180:7){3};
\node[mynode](n5) at (200:7){2};
\node[mynode](n6) at (220:6.7){1};
\node[mynode](n7) at (180:3){1};
\node[mynode](n8) at (140:2){1};
\node[mynode](n9) at (100:3){1};
\node[mynode](n10) at (70:3){1};
\node[mynode](n11) at (40:3){1};
\node[mynode](n12) at (10:3){1};
\node[mynode](n13) at (340:3){1};
\node[mynode](n14) at (310:3){1};
\node[mynode](n15) at (280:3){1};
\node[mynode](n16) at (25:4.5){1};
\node[mynode](n17) at (355:4.5){2};
\node[mynode](n18) at (325:4.5){3};
\node[mynode](n19) at (295:4.5){4};
\node[mynode](n20) at (265:4.5){5};
\node[mynode](n21) at (235:5){6};
\node[mynode](n22) at (220:3){8};
\node[mynode](n23) at (200:5){15};

\draw[myarrow](n1)--(n2);
\draw[myarrow](n2)--(n3);
\draw[myarrow](n3)--(n1);
\draw[myarrow](n4)--(n1);
\draw[myarrow](n5)--(n4);
\draw[myarrow](n6)--(n5);
\draw[myarrow](n7)--(n3);
\draw[myarrow](n8)--(n3);
\draw[myarrow](n9)--(n3);
\draw[myarrow](n10)--(n0);
\draw[myarrow](n11)--(n0);
\draw[myarrow](n12)--(n0);
\draw[myarrow](n13)--(n0);
\draw[myarrow](n14)--(n0);
\draw[myarrow](n15)--(n0);
\draw[myarrow](n16)--(n17);
\draw[myarrow](n17)--(n18);
\draw[myarrow](n18)--(n19);
\draw[myarrow](n19)--(n20);
\draw[myarrow](n20)--(n21);
\draw[myarrow](n21)--(n23);
\draw[myarrow](n0)--(n22);
\draw[myarrow](n22)--(n23);

\end{tikzpicture}

\hspace{0.5cm}
\begin{tikzpicture}

\node[mynode](n0) at (0,0){4};
%\node[mynode](n1) at (160:7){2.875};
\node[mynode](n1) at (160:7){3.375};
%\node[mynode](n2) at (180:5){2.625};
\node[mynode](n2) at (180:5){2.563}; %2.5625
%\node[mynode](n3) at (140:5){3.475};
\node[mynode](n3) at (140:5){3.469}; %3.46875
\node[mynode](n4) at (180:7){1.75};
\node[mynode](n5) at (200:7){1.5};
\node[mynode](n6) at (220:6.7){1};
\node[mynode](n7) at (180:3){1};
\node[mynode](n8) at (140:2){1};
\node[mynode](n9) at (100:3){1};
\node[mynode](n10) at (70:3){1};
\node[mynode](n11) at (40:3){1};
\node[mynode](n12) at (10:3){1};
\node[mynode](n13) at (340:3){1};
\node[mynode](n14) at (310:3){1};
\node[mynode](n15) at (280:3){1};
\node[mynode](n16) at (25:4.5){1};
\node[mynode](n17) at (355:4.5){1.5};
\node[mynode](n18) at (325:4.5){1.75};
\node[mynode](n19) at (295:4.5){1.875};
\node[mynode](n20) at (265:4.5){1.938};%1,9375
\node[mynode](n21) at (235:5){1.969};%1.96875
\node[mynode](n22) at (220:3){3};
\node[mynode](n23) at (200:5){3.484};%3.484375

\draw[myarrow](n1)--(n2);
\draw[myarrow](n2)--(n3);
\draw[myarrow](n3)--(n1);
\draw[myarrow](n4)--(n1);
\draw[myarrow](n5)--(n4);
\draw[myarrow](n6)--(n5);
\draw[myarrow](n7)--(n3);
\draw[myarrow](n8)--(n3);
\draw[myarrow](n9)--(n3);
\draw[myarrow](n10)--(n0);
\draw[myarrow](n11)--(n0);
\draw[myarrow](n12)--(n0);
\draw[myarrow](n13)--(n0);
\draw[myarrow](n14)--(n0);
\draw[myarrow](n15)--(n0);
\draw[myarrow](n16)--(n17);
\draw[myarrow](n17)--(n18);
\draw[myarrow](n18)--(n19);
\draw[myarrow](n19)--(n20);
\draw[myarrow](n20)--(n21);
\draw[myarrow](n21)--(n23);
\draw[myarrow](n0)--(n22);
\draw[myarrow](n22)--(n23);

\end{tikzpicture}
}
\caption{\textbf{Potential voting weight} (left) vs. \textbf{expected voting weight} (right)}
\label{fig:potential_expected}
\end{figure}

\paragraph{Most Influential Agents and Gurus.}

In the context of liquid democracy in the real world, graphs are characterized by a combination of direct and transitive delegations. In a LiquidFeedback-style implementation characterized by the aforementioned properties, the agents that exert the most significant influence are those that possess a high number of direct delegations, and to a lesser extent, those within a short distance of agents receiving multiple direct delegations. This is due to the fact that, under the aforementioned assumptions, transitivity induces an inherent exponential dampening of influence. It is noteworthy that these agents may or may not be endpoints, which are sometimes referred to as "gurus". The designation of an endpoint is an incidental attribute that offers scant insight into the influence of an agent.

It stands to reason that any influential agent would be prudent to establish an outgoing delegation as a precaution against the possibility of forgetting to act directly, and given that it is plausible to expect some correlation between activity and influence, the probability that an endpoint is actually influential decreases even further.

%\vspace{1cm}

%\subsection{Power Persistence Index}
%\textbf{Perhaps we can define an index based on the expected number of votes?}

%\dvd{
%Some further thoughts to better understand the new index:
%\begin{itemize}
%    \item The delegation graph with error $p$ defines a stochastic matrix with diagonal $p$ and only one off-diagonal non-$0$ entry of value $1-p$. This is a simple Markov Chain.
%    \item The standard weight ($p = 0$) corresponds to the stationary distribution of that chain with all weight flowing to endpoints and cycles (a sort of 'eventual weight').
%    \item Andreas' expectation-based weight corresponds to the following concept in Markov chains: the weight of a point corresponds to the probability that weight flows from all other points to that point via shortest paths (probability at shortest hitting time for the node or first-passage probability of the node). I believe there is a simple proof for this ... I have to write it down.
%\end{itemize}
%}

%%%%%%%%%%%%%%%%%%%%%%%%%%%%%%%%%%%%%%%%%%%%%%%%

%\newpage
\section{Outlook} \label{sec:outlook}

\paragraph{Conclusions.} By giving center stage to some of the original core ideas underpinning liquid democracy, we argued for the repositioning of research on transitive delegation systems within a broad framework of adaptive representation. To show the usefulness of this adaptive-process perspective on liquid democracy, we also elaborated one novel idea for the quantification of influence in a transitive delegation systems. 
The proposed influence measure has subsequently been shown to coincide with Random Walk Decay centrality in the model of suspendible delegations \cite{grossi2026centrality}.
The new model shows how transitivity, rather than constituting a risk towards undue accumulation of power, can actually work effectively as a mechanism to keep the excessive accumulation of power at bay.

\paragraph{Adaptive Representation: Towards a Research Agenda.} We regard this paper only as an initial step towards the development of solid theoretical and, hopefully, empirical foundations for the study of liquid democracy as an adaptive representation system. We conclude by outlining a number of open lines of research, which we find of particular interest. We stress that we do not consider such a list exhaustive.

\begin{description}

    \item[Influence by delegation.] In this paper we proposed a novel model of influence in liquid democracy that incorporates an element of adaptivity. The technical appendix has developed a first analysis of the model, but much work on it is still needed to provide a comprehensive adaptivity-based theory of influence in liquid democracy. 
    In particular, subsequent work has shown that the expected voting weight developed here coincides, in the model of suspendible delegations, with Random Walk Decay centrality, and has provided an axiomatic and computational characterization of the resulting power measure \cite{grossi2026centrality}.\footnote{We are indebted to Giorgos Papasotiropoulos for pointing out to us the relevance of the network centrality literature for the conceptualization of influence in liquid democracy.}
    %In particular, the mathematics of the model could be directly related to the concept of viscous democracy \cite{boldi2011viscous} and to centrality measures in network theory such as Katz centrality index \cite{katz1953new}.\footnote{We are indebted to Giorgos Papasotiropoulos for pointing out to us the relevance of the network centrality literature for the conceptualization of influence in liquid democracy.} 
    Similarly, it would be important to understand how our notion of influence in liquid democracy relates to approaches based on power indices from voting theory \cite{zhang2021power,colley2023measuring}.
 
    \item[Delegation hierarchies.] As we already stressed in Section \ref{sec:LF}, delegation layers can overlay, thereby describing representation patterns that refine themselves from from broad topic areas down to specific issues. In the case of LiquidFeedback, such hierarchical structures comprise default delegations, policy area-specific defaults, and topic-specific delegations. Together, they facilitate the emergence of domain-specific representation patterns that guide how new decisions are approached. This architecture allows agents to delegate once at a general level, apply their preferences broadly, and refine them incrementally over time. Such a mechanism reduces the burden on agents, thereby enhancing the operational feasibility of liquid democracy in practice. By selecting a single default delegation, an agent can ensure the presence of a fallback mechanism for any future deliberation or decision-making scenario. Additionally, agents may choose to restrict the scope of a delegation to a specific policy area or even to an individual issue. In LiquidFeedback, when multiple levels of delegation apply, the most specific applicable delegation takes priority automatically. Importantly, the presence of a delegation does not constrain an agent’s capacity to participate directly at any time. The domain-specificity of the representation patterns that can emerge through hierarchical delegations is a feature that clearly differentiates liquid democracy from other forms of representation, from representation by voting, to representation by sortition \cite{valsangiacomo2025european}. Yet, we currently lack models that can robustly account for this important feature of liquid democracy.
    
%    \item[Automatic suspension.] When it comes to such direct activity, automatic suspension is a key design feature whereby an existing delegation is temporarily overridden whenever an agent acts directly—such as by discussing, supporting, or voting—without requiring the agent to explicitly recall their delegation. This mechanism preserves individual autonomy by allowing agents to intervene directly at any point, while at the same time retaining all non-recalled delegations for future use.
    
    \item[Delegation evolution.] In the prototypical liquid democracy model, as exemplified by the LiquidFeedback implementation \cite{behrens2014principles}, agents are cognizant of the decisions made by their proxies. Dissatisfaction with these choices can prompt a reevaluation process. Agents' responses to events are reflected in the delegation graph, thereby introducing a temporal dimension \cite{behrens2022temporal} to liquid democracy. This temporal dimension tracks the adaptivity of the system over time and could be interpreted as a collective form of learning. This fluidity of delegations can enable the development of epistemic specialization, where agents build reputations in particular domains—such as environmental policy, finances, traffic, etc.—and receive delegated authority accordingly. The ability to make decisions that are both responsive and accurate is directly linked to maintaining influence in future situations of a similar nature. The use of learning theory---broadly conceived---within social choice and collective decision making is currently in its infancy \cite{armstrong2019machine},\footnote{See the SCaLA workshop series (Workshop on Social Choice and Learning Algorithms), \url{https://sites.google.com/view/scala25/}} but we believe it to be of special relevance for the analysis of the adaptivity of representation patterns in liquid democracy. 
\end{description}

We conclude with a general methodological point on the use of models as tools for analysis and, ultimately, intervention, which we believe should provide perspective on any modeling effort in liquid democracy research. The development of models for liquid democracy can, ideally, validate existing implementations of the system (as in LiquidFeedback), or suggest clear pathways for improving existing solutions. However, models need then to incorporate---albeit abstractly---all the relevant features of the deployment context. As previously exemplified in relation to transitive delegations and power accumulation (Section \ref{par:accumulation}), a too simplistic model of influence (what we referred to as potential weight), may suggest that limiting the number of allowed delegation steps or imposing costs on delegation transfers can constitute desirable interventions. Yet, in light of a better-informed model (like the expected weighted model we proposed here), such interventions may yield the opposite effect by  inadvertently encouraging agents to bypass intermediary layers and delegate directly to perceived influential voters. This behavior can lead to the aggregation of direct delegations and, paradoxically, foster the accumulation of less constrained power—precisely the outcome such measures aim to mitigate. If models are to be used to support the design of interventions, it is essential to carefully consider their ability to capture the core features of the problem that are relevant for practice.

\bigskip

\appendix

\medskip
\noindent
{\Large {\bf Technical Appendix}}
\medskip

\noindent
This technical appendix provides a more extensive mathematical presentation of the key ideas introduced in the main body of the paper.

\section{Preliminaries: A Simple Model of Delegations}

We start with a very simple model of delegations in liquid democracy \cite{christoff2017binary,brill2018pairwise,kahng2021liquid} that has been predominant in the literature so far, with the exception of \cite{brillcycles}.

\subsection{Delegations} 
Let $\N = \set{1, \ldots, n}$ be the set of users/voters. A {\em delegation profile} is a function $\D: \N \to \N$, where $\D(i)$ denotes $i$'s proxy. For $X \subseteq \N$, $\D(X) = \set{i \in \N \mid \exists j \in X: D(j) = i}$ denotes the set of images of elements of $X$, that is, the set of proxies of users in $X$. If $X = \D(X)$ and there exists no $Y$ such that $Y \subset X$ and $Y = \D(Y)$ we say that $X$ is ultimate. If $\set{i}$ is ultimate, then we refer to $i$ as an endpoint voter. Clearly, $i$ is an endpoint if and only if $D(i) = i$. We denote by $\D^\star(i)$ the closest (by length of delegation path) user to $i$ that belongs to an ultimate set. When $\D^\star(i)$ belongs to an ultimate set which is a singleton then, intuitively, $\D^\star(i)$ is the ultimate proxy who votes for $i$. When $\D^\star(i)$ belongs to an ultimate set that is not a singleton, then $i$ delegates into a cycle. We furthermore denote by $\N^\star$ the set of endpoints of $\D$. 

If $X$ is ultimate, but it is not a singleton, then any element of $X$ is the image of an $|X|$-fold iteration of $D$, that is: $i = \D^{|X|}(i)$. In other words, there exists a delegation cycle that links all users in $X$.   

We will refer to delegation profiles also as {\em delegation graphs} where $i \to j$ if and only if $\D(i) = j$. So, delegation graphs are directed graphs over $\N$. Structurally, they consist of trees whose roots lay on points in cycles, that is, trees whose roots belong to ultimate sets.

\subsection{Potential Voting Weight} 
Delegation profiles represent the allocation of voting rights at the very moment in which ballots are counted after voting. So, given $\D$, the (potential) voting weight that $i$'s ballot carries is:
\begin{equation} \label{eq:naive}
    w_{\D}(i) = 1 + \sum_{j \in \N} \1_{D^\star(j) = i},
\end{equation}
that is, $i$'s weight (i.e., $1$) plus the weight of those that would use $i$ as ultimate proxy in $D$ (i.e., $\sum_{j \in \N} \1_{\D^\star(j) = i}$) if $i$ were to vote given the delegations specified in $\D$. This is the notion of weight illustrated on the left-hand side of Figure \ref{fig:potential_expected}.
%If there are ultimate sets in $X$ that are not endpoints, then $\sum_{i \in \N} w_\D(i) < n$ as some voting weight is lost through delegation cycles. 

%%%%%%%%%%%%%%%%%%%%%

\section{A Model of Influence When Delegations Are Suspendible} 
%\an{Given that it's a Latin word, should it be "suspendible"?}

We now extend the above simple model of delegations to capture the process-centered feature of liquid democracy according to which users select proxies as defaults to direct participation, should the user decide not to vote. The result is a model of delegation that incorporates explicitly the possibility of transitioning between two modes of representation: representation via direct voting, vs. representation via proxy.

\subsection{Suspendible Delegations}

To do this we extend delegation graphs with a function that assigns to all users, who are not endpoints, a probability $p_i$ corresponding to the probability that $i$ casts a ballot directly, instead of delegating to $D(i)$ (which therefore happens with probability $1 - p_i$). The endpoints of $\D$ are assumed to cast a direct ballot with probability $1$. 

\begin{definition}[Suspendible delegation profile]
    A suspendible delegation profile is a tuple $\CalD = \tuple{\N,\D, p}$ where:
    \begin{itemize}
        \item $\N$ is a set of users/agents;
        \item $\D$ is a delegation profile over $\N$;
        \item $p : \N \backslash \N^\star \to [0,1]$, that is, $p$ assigns to each agent a probability of voting directly.
    \end{itemize}
    We will often write $p_i$ instead of $p(i)$.
\end{definition}
A suspendible delegation profile therefore fully describes the probability $\P(ij)$ of any given delegation $i \to j$ to take place, as follows:
\begin{equation} \label{eq:d_matrix}
    \P(ij) = 
    \begin{cases}
        1  & \text{if   } D(i)=j, i = j \\
        1- p_i & \text{if   } D(i)=j, i \neq j \\
        p_i & \text{if   } D(i) \neq j, i = j \\
        0 & \text{otherwise} 
    \end{cases}
\end{equation}

Equation \eqref{eq:d_matrix} defines an $n \times n$ stochastic matrix $[\P(ij)]$ (called {\em delegation matrix}) where each row has value $1$ or $p_i$ on the diagonal entry, and at most one off-diagonal entry with value $1 - p_i$. Equivalently, the matrix defines a weighted directed graph in the standard way. In what follows we will sometimes refer to the delegation matrix $[\P(ij)]$ of $\CalD$ also simply as $\CalD$. 

\medskip

There are two natural ways in which to interpret $\CalD$ and its matrix $[\P(ij)]$.
\begin{description}
    \item[Lottery on delegation graphs.] Each $\CalD = \tuple{\N, \D, p}$ is equivalent to a lottery over the set of all delegation profiles $\N^\N$. Such a lottery is the product probability distribution obtained from all Bernoulli trials $p_i$ with outcomes $\set{i, \D(i)}$ where $i$ is not an endpoint of $\D$. So, the probability of sampling from $\CalD$ an arbitrary profile $\D' \in \N^\N$ is 
\begin{equation}
    \P(\D') = \prod_{ij \in \D'} \1_{ij \in \D}\P(ij). 
\end{equation}
    \item[Flow of voting weight.] Matrix $[\P(ij)]$ can be viewed to describe the probability with which the voting right of a user ``flows'' to another user via delegations. The matrices defined by Equation \eqref{eq:d_matrix} are instances of the class of weighted delegation profiles as studied in the liquid democracy literature \cite{zhang2022tracking} as well as of command games \cite{hu2003authoritya,hu2003authorityb} from the cooperative game theory literature.
\end{description}

%%%%%%%%%%%%%%%%%%

\subsection{Influence as Expected Weight}

We first explore the implications for the understanding of power with suspendible delegations following the first one of the two interpretations above. 

We can think of the voting weight of a given user $i$ as a random variable determined by the lottery $\CalD$ defined over all possible delegation graphs. So, $i$'s influence in $\CalD$ can be thought of as the expectation on $i$'s weight when the delegation graph is sampled from $\CalD$ ({\em expected weight}): 
\begin{equation}
    \phi_i = \E_{D \sim \CalD}[w_i] = \sum_{\D \in \N^\N} \P(\D)\cdot w_i(\D). \label{eq:pow1}
\end{equation}
It is worth noticing that this is a special case of the notion of power in weighted delegation graphs studied in \cite{zhang2022tracking}.

\begin{example}
    Recall Example \ref{ex:direct}. We will recompute the value obtained earlier, but now following Equation \eqref{eq:pow1}. In that example, the suspendible delegation profile defines a lottery over delegation graphs where each graph corresponds to the success of $m$ specific trials out of $10000$. The number of delegations reaching the endpoint voter is therefore distributed binomially with parameters $10000$ (trials) and $0.5$ (probability of success). The mean of such distribution is the number of trials times the probability of success: $5000$. As all graphs contain the endpoint voter, its vote is counted with probability $1$. This gives us the expected weight of $5000 + 1$. 
\end{example}

%%%%%%%%%%%%%%%%%%%%%%%

\subsection{Influence as Flow}

We now explore the second interpretation and will show it leads to an equivalent definition to Equation \eqref{eq:pow1}. We think of the flow of voting weight in the delegation graph as a stochastic process $(X_k)_{0 \leq k}$ where $X_k$ denotes the state (i.e., user) to which the weight flows over time starting from state/user $X_0$. In other words, we view the flow of power in a suspendible delegation profile as a Markov chain.

\paragraph{Stationary weight distributions of suspendible delegations.}

The stationary distribution of the stochastic process $(X_k)_{0 \leq k}$ defined by a delegation matrix $\CalD = [\P(ij)]$ is a probability distribution $p: \N \to [0,1]$ ($\sum_{i \in \N} p(i) = 1$) such that:
\begin{equation}
    p \CalD = p, \label{eq:stationary0}
\end{equation}
that is,
\begin{equation}
    \sum_{j \in N} p(i) \P(ji) = p(i). \label{eq:stationary}
\end{equation}

What Equation \eqref{eq:stationary} intuitively captures is a form of ultimate state of the distribution of voting power that keeps delegations `in equilibrium' in the sense that the voting weight allocated to each user would not change any more---would not `flow' somewhere else---as an effect of delegations. 
In fact in the literature on command games such an equation is referred to as {\em authority equilibrium equation} \cite{hu2003authorityb}. It follows that such an equilibrium weight allocation corresponds actually to the allocation of voting weight that ultimate nodes (be them endpoints, or members of a delegation cycle) accrue via delegations, assuming the weight is distributed equally across all users (at time $0$). It incorporates a principled way to think of the weight that remains present in delegation cycles when these are not broken by a member of the cycle voting directly. In this sense, it provides a more fine-grained formalization of a conception of weight in liquid democracy proposed in the literature, which interprets the weight of cycles as null (see, e.g., \cite{christoff2017binary}).

%Importantly, such conception of weight is different from the allocation by expected weight of Equation \eqref{eq:pow1} determines.

\medskip

General results in the theory of Markov chains guarantee that the stationary distribution of Equation \eqref{eq:stationary} exist for any given suspendible delegation profile.\footnote{Essentially because every diagonal entry of matrix $[\P(ij)]$ is non-zero and the matrix is therefore aperiodic.}
We can then use Equation \eqref{eq:stationary} to determine, for each $\CalD$ what such a stationary weight distribution would be, starting from the equal allocation of weight where each user has a fraction $\frac{1}{n}$ of total weight.

\begin{fact}
    Let $\CalD$ be a delegation matrix. The stationary distribution $p$ of $\CalD$ (assuming a uniform initial distribution $[\frac{1}{n} \ldots \frac{1}{n}]$)\footnote{This assumption is needed (although any distribution would be possible) because, in general, $\CalD$ may typically consist of multiple disconnected components and is therefore reducible.} is such that:
    \begin{itemize}
        \item $p(i) = \frac{\sum_{j \in \N} \D^*(j)}{n}$, if $i$ is an endpoint;
        \item $p(i) = \frac{\sum_{j \in U}\sum_{k \in \N} \D^*(k)}{n}\cdot\frac{1}{|U|}$, if $i$ belongs to an ultimate set $U$ but is not an endpoint;
        \item $p(i) = 0$, otherwise.
    \end{itemize}
\end{fact}
\begin{proof}
First, let us assume $\CalD$ consists of one component. In such a case, the three statements reduce to:
    \begin{itemize}
        \item $p(i) = 1$, if $i$ is an endpoint;
        \item $p(i) = \frac{1}{|U|}$, if $i$ belongs to an ultimate set $U$ but is not an endpoint;
        \item $p(i) = 0$, otherwise.
    \end{itemize}
Let us reason by cases. If $i$ is an endpoint of $\CalD$, then the vector where for all $j$ $D^*(j) = i$, $p(j) = 0$, and where $p(i) = 1$ trivially satisfies Equation \eqref{eq:stationary}.
%restricted to the component of $\CalD$ that includes $i$. 
If $i$ belongs to an ultimate set $U$ but is not an endpoint, then $U$ consists of a cycle one member of which being $i$. For each element $k$ of $U$ we consider the number $\sum_{k \in \N} \D^*(k)$ of users delegating (via a shortest path) to $k$. So the total number of users with a path into $U$ is $\sum_{j \in U}\sum_{k \in \N} \D^*(k)$, which we normalize to $1$. 
%Let us know consider the component of $\CalD$ containing $U$. 
Then, the vector where for all $j$ such that $D^*(j) = i$, $p(j) = 0$, and where $p(i) = \frac{1}{|U|}$ can also be seen to satisfy Equation \eqref{eq:stationary}.
%restricted to the component of $\CalD$ that includes $U$. 
Finally, if $i$ does not belong to an ultimate set, by the structure of $\CalD$ it is connected to some ultimate set. The earlier two cases then establish also the last claim.

If $\CalD$ consists of more than one component, and we assume that the initial state of the stochastic process (initial voting weight) is distributed uniformly at random, then the ultimate set $U$ of each component is reached with probability equal to the size of the component divided by $n$. We thus obtain the values in the statement of the theorem.  
\end{proof}

\begin{example}
Figure \ref{fig:potential_stationary} illustrates the difference between potential weight and stationary weight (cf. Figure \ref{fig:potential_expected}). Notice that in the stationary weight case every member of the cycle recieves $\frac{1}{3}$ of the total weight ($9$) while the remaining agents have weight $0$.
\end{example}

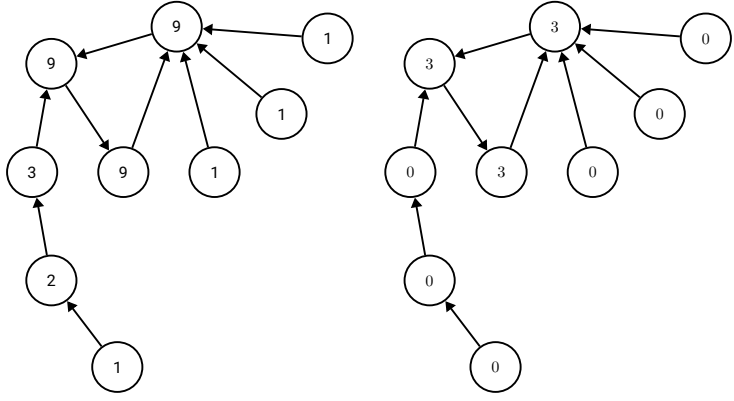
\begin{figure}[t]
\begin{center}
\scalebox{0.6}{
\begin{tikzpicture}

\node[mynode](n1) at (160:7){9};
\node[mynode](n2) at (180:5){9};
\node[mynode](n3) at (140:5){9};
\node[mynode](n4) at (180:7){3};
\node[mynode](n5) at (200:7){2};
\node[mynode](n6) at (220:6.7){1};
\node[mynode](n7) at (180:3){1};
\node[mynode](n8) at (140:2){1};
\node[mynode](n9) at (100:3){1};

\draw[myarrow](n1)--(n2);
\draw[myarrow](n2)--(n3);
\draw[myarrow](n3)--(n1);
\draw[myarrow](n4)--(n1);
\draw[myarrow](n5)--(n4);
\draw[myarrow](n6)--(n5);
\draw[myarrow](n7)--(n3);
\draw[myarrow](n8)--(n3);
\draw[myarrow](n9)--(n3);

\end{tikzpicture}
\hspace{0.5cm}
\begin{tikzpicture}

\node[mynode](n1) at (160:7){$3$};
\node[mynode](n2) at (180:5){$3$};
\node[mynode](n3) at (140:5){$3$};
\node[mynode](n4) at (180:7){$0$};
\node[mynode](n5) at (200:7){$0$};
\node[mynode](n6) at (220:6.7){$0$};
\node[mynode](n7) at (180:3){$0$};
\node[mynode](n8) at (140:2){$0$};
\node[mynode](n9) at (100:3){$0$};

\draw[myarrow](n1)--(n2);
\draw[myarrow](n2)--(n3);
\draw[myarrow](n3)--(n1);
\draw[myarrow](n4)--(n1);
\draw[myarrow](n5)--(n4);
\draw[myarrow](n6)--(n5);
\draw[myarrow](n7)--(n3);
\draw[myarrow](n8)--(n3);
\draw[myarrow](n9)--(n3);

\end{tikzpicture}
}
\end{center}
\caption{\textbf{Potential voting weight} (left) vs. \textbf{stationary voting weight} (right)}
\label{fig:potential_stationary}
\end{figure}

%%%%%%%%%%%%

\paragraph{First-passage probabilities of suspendible delegations.}

While Equation \eqref{eq:stationary} captures the equilibrium state of a weight allocation in a suspendible delegation profile, we are now interested instead in the probability that the weight of a given user $i$ reaches a given user $j$ via the delegation path (if any) that connects the two. This amounts to determining the so-called first-passage probabilities with which voting weight of each user reaches any other user. 

\medskip

The probability that $j$'s weight flows to $i$ via the shortest path $i \to i_1 \to \ldots \to i_{\ell} = j$ is 
\begin{equation}
    p(ji) = \prod_{1 \leq k \leq \ell} (1 - p(i_k))
\end{equation}
We can now think of the voting weight of a user $i$ as their own weight plus the sum of the first-passage probabilities from any other user $j$ to $i$, that is:
\begin{equation}
    \psi_i = 1+ \sum_{j \in \N} p(ji).
\end{equation}

This notion of influence is illustrated in the right-hand side of Figure \ref{fig:potential_expected}.

\subsection{Equivalence}

It is not difficult to show that the notion of weight based on first-passage probability is equivalent to the earlier one we introduced in the appendix, based on lotteries on delegation graphs (Equation \eqref{eq:pow1}). Intuitively, the weight that a user should expect based on the delegation paths made possible by a suspendible delegation profile is equivalent to the probability in which weight flows to the user from any other user in the profile.
\begin{theorem}
    For any suspendible delegation profile $\CalD$ and user $i$:
    \begin{equation}
        \phi_i = \psi_i.
    \end{equation}
\end{theorem}
\begin{proof}
The claim is proven by the following series of equivalences:
\begin{align*}
    \phi_i & = \sum_{\D \in \N^\N} \P(\D')\cdot w_i(\D) \\
    & = \sum_{\D \in \N^\N} \P(\D)\cdot (1 + \sum_{j \in \N} \1_{D^\star(j) = i}) & \mbox{by Equation \eqref{eq:naive}} \\
    & = \sum_{\D \in \N^\N} \P(\D) \cdot  1 + \sum_{\D \in \N^\N} \P(\D) \cdot \sum_{j \in \N} \1_{D^\star(j) = i} \\
    & = 1 + \sum_{\D \in \N^\N} \P(\D) \cdot \sum_{j \in \N} \1_{D^\star(j) = i} \\
    & = 1 + \sum_{j \in \N} p(ji) \\
    & = \psi_i.
\end{align*}
The second-last equivalence holds for the following reason. We consider the sum over the potential weight of $i$ in an arbitrary graph $D'$, that is, the number of delegation paths linking some user $j$ to $i$, weighted by the probability of $D'$ (according to the given suspendable delegation profile $\CalD$). This amounts to the sum over the probability of each such path from $j$ to $i$, which in turn is precisely the sum over all first-passage probabilities from any $j$ to $i$.
\end{proof}

\subsection{Properties of Power When Delegations Are Suspendible}

The following fact formalizes and generalizes some simple observations made in Section \ref{sec:expected}.

\begin{fact}
Let $\CalD = \tuple{\N, \D, p}$ be a suspendible delegation profile such that for all $i,j \in \N$, $p_i = p_j = p$. Then:
\begin{itemize}
    \item for any $i \in \N$, $\lim_{p \to 0} \E_{\D \sim \CalD}[w_i] = w_i(\D)$.
    \item for any $i \in \N$,  $\lim_{p \to 1} \E_{D \sim \CalD}[w_i] = 1$.
    \item for any $i \in \N$, any delegation path of length $\ell$ to $i$ contributes
    \begin{equation}
        \sum_{1 \leq k \leq \ell} p^k = p \frac{ (1 - p^\ell)}{1 - p}
    \end{equation}
    to $\E_{D \sim \CalD}[w_i]$. Furthermore, as $\ell \to \infty$ the contribution of the path to $i$'s expected weight tends to $\frac{p}{1-p}$ and is therefore upper-bounded by $1$.
\end{itemize}
\end{fact}
It follows that expected weight grows sublinearly in the length of delegation paths, unlike in the case of potential weight, and can never yield more than one extra vote per path in expectation. As a consequence, we can make the following simple but insightful observation.
\begin{corollary}
    Let $\CalD = \tuple{\N, \D, p}$ be a suspendible delegation profile. Then for any $i \in \N$
    \begin{equation}
        \E_{\D \sim \CalD}[w_i] \leq f_\CalD(i),
    \end{equation}
    where $f_\CalD(i)$ is the number of maximal delegation paths reaching $i$ in $\CalD$.
\end{corollary}
In other words, expected weight can be bounded by the total number of maximal (that is, that cannot be extended further) delegation paths pointing to a user, independently of their length. This stands in stark contrast with the behavior of potential weight, for which such a bound does not hold.

\addcontentsline{toc}{section}{Acknowledgment}

\section*{Acknowledgments}

%Funded by the European Union. Views and opinions expressed are however those of the author(s) only and do not necessarily reflect those of the European Union or European Research Executive Agency. Neither the European Union nor the granting authority can be held responsible for them.

The authors acknowledge support by the European Union under the Horizon Europe \href{https://perycles-project.eu/}{Perycles} project (Participatory Democracy that Scales).

\bigskip
%\begin{figure}[b]
\begin{center}
    \includegraphics[width=0.5\textwidth]{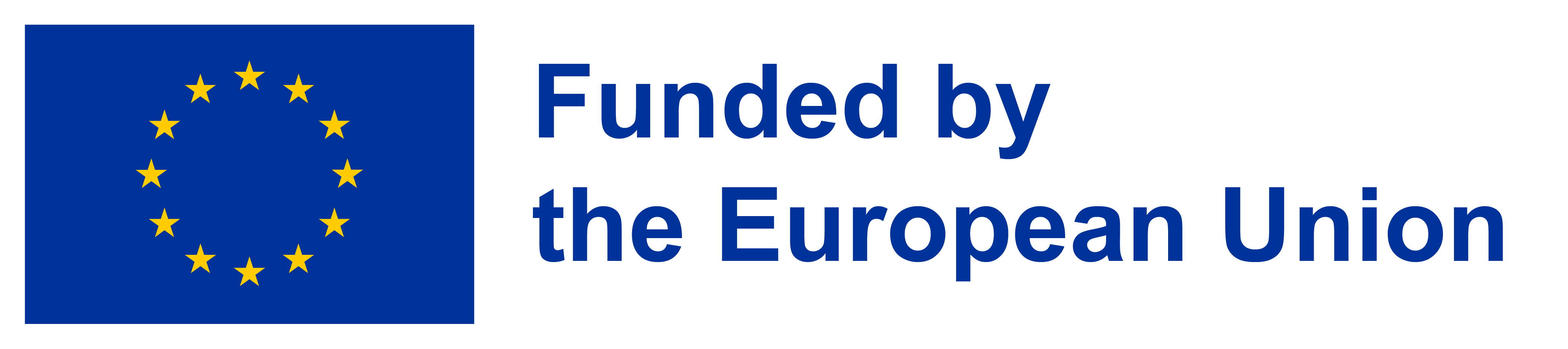}
\end{center}
%\end{figure}

\addcontentsline{toc}{section}{References}
\bibliographystyle{plain}
%\bibliography{lit}

\end{document}